\documentclass{article}
\usepackage{graphicx} 

\usepackage{latexsym}
\usepackage[english]{babel}
\usepackage{amsmath}
\usepackage{amsfonts}
\usepackage{amssymb}
\usepackage{makeidx}
\usepackage[utf8]{inputenc}
\usepackage{lmodern}
\usepackage{verbatim}
\usepackage{cite} 
\usepackage[pagewise]{lineno}
\usepackage{amsthm}
\usepackage{tikz-cd} 

\newtheorem{theorem}{Theorem}[section]

\theoremstyle{definition}
\newtheorem{definition}[theorem]{Definition}
\newtheorem{remark}[theorem]{Remark}
\newtheorem{example}[theorem]{Example}

\newcommand*{\Reeb}{\mathcal{R}}
\newcommand{\R}{\ensuremath{\mathbb{R}}}

\newcommand{\Flder}{\rightarrow}

\usepackage{amssymb}

\newcommand{\proa}{A^*G \mbox{$\;$}_{\tau^*} \kern-3pt\times_\alpha
G \mbox{$\;$}_\beta \kern-3pt\times_{\tau^*} A^*G}

\newcommand{\al}{\mathfrak{g}}

\title{a}
\author{colombo2 }
\date{March 2023}

\begin{document}
	\title{Jacobi structure for dissipative mechanical systems on Lie Algebroids}
	\author{
		{\bf\large Alexandre Anahory  Simoes$^{1}$, Leonardo Colombo$^{2}$,}\\
		{\bf\large Manuel de León$^{3}$},{\bf\large Modesto Salgado$^{4}$ and Silvia Souto$^{5}$}\\ 
		{\it\small $^{1}$ IE School of Science and Technology, Madrid, Spain }\\
        {\it\small $^{2}$ Centre for Automation and Robotics (CSIC-UPM),}
        {\it\small Arganda del Rey, Spain}\\ 
        {\it\small$^{3}$Instituto de Ciencias Matematicas  (CSIC), Madrid, Spain}\\ {\it\small$^{3}$Real Academia de Ciencias, Madrid, Spain}  \\
        {\it\small $^{4, 5}$ Departamento de Matem\'aticas, Facultade de Matem\'aticas,} \\
        {\it\small Universidade de Santiago de Compostela, Spain} \\
        {\it\small $^{4, 5}$ Centro de Investigaci\'on y Tecnolog\'ia Matem\'atica de Galicia (CITMAga), Spain}\footnote{Corresponding author: $^{1}$ alexandre.anahory@ie.edu}}

\maketitle

\begin{abstract}
    We extend the Jacobi structure from $TQ\times \R$ and $T^{*}Q \times \R$ to $A\times \R$ and $A^{*}\times \R$, respectively, where $A$ is a Lie algebroid and $A^{*}$ carries the associated Poisson structure. We see that $A^*\times \R$ possesses a natural Jacobi structure from where we are able to model dissipative mechanical systems, generalizing previous models on $TQ\times \R$ and $\mathfrak{g}\times \R$.
    
\end{abstract}

\section{Introduction}
 In the last years there has been a lot of applications of Lie algebroids in
theoretical physics and other related sciences, more precisely in Classical Mechanics, Classical Field Theory and their applications. The main
point is that Lie algebroids provide a general framework for systems with different features as systems with symmetries, systems over semidirect products,
Hamiltonian and Lagrangian systems, systems with constraints (nonholonomic and vakonomic), higher-order mechanics and optimal control \cite{colombo2,Pepin2007,PepinEduardo,colombo,CoLeMaMa,CoMa,januzkatrina,MdLMdDJCM,JCsolo, MaRoPa, Eduardofieldtheory, Eduardo, Eduardo1, Eduardoalg, Eduardoho, Ma2, Tom} among many others. 

In \cite{LMM} M. de Le\'on, J.C Marrero and E. Mart\'inez have
developed a Hamiltonian description for the mechanics on Lie
algebroids and they have shown that the dynamics is obtained solving
an equation for the Hamiltonian section (Hamiltonian vector field) in the same way than in Classical Mechanics 
(see also \cite{Eduardo1}). Moreover, for a Lie algebroid $A$, they have shown that the
Legendre transformation $\mathbb{F}L:A\to A^{*}$ associated to a 
Lagrangian $L:A\to\R$ induces a Lie algebroid morphism and when the
Lagrangian is regular both formalisms are equivalent. Marrero and collaborators also have studied non-holonomic mechanics on Lie algebroids \cite{CoLeMaMa}. In
other direction, in \cite{IMMS} D. Iglesias, J.C. Marrero, D.
Mart\'in de Diego and D. Sosa have studied singular Lagrangian
systems and vakonomic mechanics from the point of view of Lie
algebroids obtained through the application of a constrained
variational principle. 

Contact Hamiltonian and Lagrangian systems have deserved a lot of attention in recent years (see  \cite{anahory2021geometry}, \cite{Bravetti2017}, \cite{Bravetti2018}, \cite{deLeon2018}, and references therein).
One of the most relevant features of contact dynamics is the absence of conservative properties
contrarily to the conservative character of the energy in symplectic dynamics; indeed, 
we have a dissipative behavior.

The outline of the paper is the follwoing: in Section 2 we review the definition of the Hamiltonian equations in the canonical Poisson structure of $A^{*}$, where $A$ is a Lie algebroid. In Section 3, we introduce a Jacobi structure on $A^{*}\times \R$, from where we derive the equations of motion for a contact Hamiltonian system and use it to obtain the contact Lagrangian systems on the bundle $A\times \R$. In Section 4, we particularize our construction to the case of the Atiyah algebroid and we derive, from a geometric viewpoint, the contact Wong's equations. 

\section{Mechanics on Lie algebroids}

Let $\tau_{A}:A\rightarrow Q$ be a vector bundle  with base manifold $Q$ together with a fiber-preserving map $\rho:A\rightarrow TQ$ called the \textit{anchor} map. Let $[\cdot ,\cdot]_{A}:\Gamma(A)\times\Gamma(A)\rightarrow\Gamma(A)$ be a Lie bracket on the set of sections $\Gamma(A)$ (that is, a skew-symmetric bilinear map satisfying the Jacobi identity) satisfying the Leibniz rule
\begin{equation*}
[X ,fY]_{A}=f[X ,Y]_{A}+\rho(X)(f)Y, \ \ \text{for} \ X,Y\in \Gamma(A) \ \text{and} \ f\in C^{\infty}(Q).
\end{equation*}

The vector bundle $\tau_{A}:A\rightarrow Q$ equipped with the anchor map $\rho$ and the bracket of sections $[\cdot ,\cdot]_{A}$ is called a \textit{Lie algebroid}  (see \cite{Mack} for instance). 





Suppose that $(q^{i})$ are local coordinates on $Q$ and that $\{e_{a}\}$ is a local basis of sections of $A$. The local functions $C^{d}_{a b}$ and $\rho^{i}_{a}$ defined by
\begin{equation*}
[e_{a} , e_{b}]_{A}=C^{d}_{a b} e_{d}, \quad \rho(e_{a})=\rho^{i}_{a}\frac{\partial}{\partial q^{i}},
\end{equation*}
are called \textit{structure functions} of the Lie algebroid. In the following sections, we will also considered local coordinates $(q^{i},y^{a})$ on $A$, adapted to the local basis of sections $\{e_a\}$. 

We will see that the Lie algebroid structure on $A$ naturally induces a Poisson structure on its dual bundle $A^{*}$. First, we will review the definition of a linear Poisson structure on the vector bundle $\pi_{A^{*}}:A^{*}\rightarrow Q$.

\begin{definition}
	A \textit{linear Poisson structure} on $A^{*}$ is a bracket of functions $\{\cdot,\cdot\}_{A^{*}}:C^{\infty}(A^{*})\times C^{\infty}(A^{*}) \rightarrow C^{\infty}(A^{*})$ such that:
	\begin{enumerate}
		\item $\{\cdot,\cdot\}_{A^{*}}$ is a skew-symmetric bilinear map satisfying the Jacobi identity;
		\item $\{\cdot,\cdot\}_{A^{*}}$ satisfies the Leibniz rule
		\begin{equation*}
		\{f,g h\}_{A^{*}}=g	\{f, h\}_{A^{*}}+	\{f,g\}_{A^{*}}h, \quad f,g,h\in C^{\infty}(A^{*}).
		\end{equation*}
		\item $\{f,g\}_{A^{*}}$ is a fiberwise linear function if $f$ and $g$ are also fiberwise linear functions on $A^{*}$.
	\end{enumerate}
\end{definition}

The Poisson bracket is associated with a bi-vector $\Lambda_{A^{*}}$ on $A^{*}$ such that
\begin{equation*}
\{f,g\}_{A^{*}}=\Lambda_{A^{*}}(df,dg), \quad \forall f,g \in C^{\infty}(A^{*}).
\end{equation*}

There is a correspondence between sections of $A$ and fiberwise linear functions on $A^{*}$. If $X\in \Gamma(A)$ then the corresponding linear function is denoted by $\hat{X}:A^{*}\rightarrow \R$ and given by
\begin{equation*}
\hat{X}(\alpha_{q})=\langle\alpha_{q},X(q) \rangle,
\end{equation*}
for any $\alpha_{q}\in A^{*}_{q}$.

Now, the following theorem (see \cite{MdLMdDJCM} for the proof) contains the correspondence between Lie algebroid and linear Poisson structures.

\begin{theorem}
	There exists a one-to-one correspondence between Lie algebroid structures on the vector bundle $\tau_{A}:A\rightarrow Q$ and linear Poisson structures on the dual bundle $A^{*}$. The correspondence is determined by the relations
	\begin{equation*}
	\widehat{[X,Y]}_{A}=\{\hat{X},\hat{Y}\}_{A^{*}}, \quad \rho(X)(f)\circ \pi_{A^{*}}=\{\hat{X}, f\circ \pi_{A^{*}}\}_{A^{*}}
	\end{equation*}
	for $X, Y\in \Gamma(A)$ and $f\in C^{\infty}(Q)$.
\end{theorem}

If $\{e_a\}$ is a local basis of sections for $A\rightarrow Q$, the dual basis $\{e^{a}\}$ is a local basis of sections for the dual bundle. In local coordinates $(q^{i},p_{a})$ on $A^{*}$, adapted to the dual basis $\{e^{a}\}$, the Poisson structure associated to a Lie algebroid $\tau_{A}:A\rightarrow Q$ is given by the relations
\begin{equation}
\{p_{a},p_{b}\}_{A^{*}}=C_{a b}^{d} p_{d}, \quad \{p_{a},q^{i}\}_{A^{*}}=\rho_{a}^{i}, \quad \{q^{i},q^{j}\}_{A^{*}}=0.
\end{equation}

Having fixed a Poisson structure we can provide Hamiltonian equations as follow. Let $h:A^{*}\rightarrow \R$ be a Hamiltonian function. Then the \textit{Hamiltonian vector field} is the vector field $X_{h}$ characterized by
\begin{equation*}
X_{h}(f)=\{ h,f \}_{A^{*}}, \quad \forall \ f \in C^{\infty}(A^{*}).
\end{equation*}
Alternatively, using the Poisson structure $\Lambda_{A^{*}}$,
\begin{equation*}
X_{h}=i_{dh} \ \Lambda_{A^{*}}.
\end{equation*}
Thus, in local coordinates,
\begin{equation*}
X_{h}=\rho_{a}^{i}\frac{\partial h}{\partial p_{a}}\frac{\partial}{\partial q^{i}}-\left( \rho_{a}^{i}\frac{\partial h}{\partial q^{i}}+C_{a b}^{d} p_{d}\frac{\partial h}{\partial p_{b}} \right)\frac{\partial}{\partial p_{a}}.
\end{equation*}
Therefore, the \emph{Hamilton equations} are
\begin{equation}\label{Hameq}
  \frac{d q^i}{d t}=\rho^i_{a} \frac{\partial h}{\partial
    p_{a}},\qquad \frac{d p_{a}}{d t}=- \rho^i_{a}
  \frac{\partial h}{\partial q^i}- p_{\gamma}
  C^{\gamma}_{ab} \frac{\partial h}{\partial p_{b}}.
\end{equation}


\begin{example}
Consider a \textit{tangent bundle} of a manifold $Q.$ The sections
of the bundle $\tau_{TQ}:TQ\to Q$ are the set of vector
fields on $Q$. The anchor map $\rho:TQ\to TQ$ is the identity
function and the Lie bracket defined on $\Gamma(\tau_{TQ})$ is
induced by the Lie bracket of vector fields on $Q.$

Note that in this case, Hamilton equations \eqref{Hameq} become in the usual Hamilton equations $$\dot{q}^{i}=\frac{\partial h}{\partial
p_i},\quad\dot{p}_{i}=-\frac{\partial h}{\partial q^{i}}.$$
\end{example}
\begin{example}
 Given a \textit{finite dimensional real Lie algebra} $\mathfrak{g}$ and 
$Q=\{q\}$ be a unique point, we consider the vector bundle
$\tau_{\mathfrak{g}}:\mathfrak{g}\to Q.$ The sections of this
bundle can be identified with the elements of $\mathfrak{g}$ and
therefore we can consider as the Lie bracket the structure of the
Lie algebra induced by $\mathfrak{g}$, and denoted by $[\cdot,\cdot]_{\mathfrak{g}}$. Since
$TQ=\{0\}$ one may consider the anchor map $\rho\equiv 0$. The triple
$(\mathfrak{g},[\cdot,\cdot]_{\mathfrak{g}},0)$ is a Lie algebroid
over a point.

Note that in this case, Hamilton equations \eqref{Hameq} become in the Lie-Poisson equations $$\dot{p}_{a} = \frac{d p_{a}}{d t}=- p_{\gamma}
  C^{\gamma}_{ab} \frac{\partial h}{\partial p_{b}}.$$
\end{example}

\section{Contact mechanical systems on Lie algebroids}

Let us recall first the definition of a Jacobi structure {(see \cite{Kirillov1976} and \cite{Lichnerowicz1978}).

\begin{definition}(Jacobi structure)
	A \textit{Jacobi structure} on a manifold $M$ is a pair $(\Lambda,E)$, where $\Lambda$ is a bi-vector field and $E$ is a vector field, satisfying the following equations
	\begin{equation*}
		[\Lambda,\Lambda]=2E\wedge \Lambda, \quad [E,\Lambda]=0,
	\end{equation*}
	with $[\cdot,\cdot]$ the Schouten-Nijenhuis bracket.
\end{definition}

\begin{definition}(Jacobi bracket)
	A Jacobi bracket on a manifold $M$ is a bilinear, skew-symmetric map $\{\cdot,\cdot\}:C^{\infty}(M)\times C^{\infty}(M) \rightarrow C^{\infty}(M)$ satisfying the Jacobi identity and the following weak Leibniz rule
	\begin{equation*}
		\text{supp}(\{f,g\})\subseteq \text{supp}(f)\cap \text{supp}(g).
	\end{equation*}
\end{definition}

A \textit{Jacobi manifold} is a manifold possessing either a Jacobi structure or a Jacobi bracket since these two definitions are equivalent (see \cite{Kirillov1976}, \cite{Lichnerowicz1978}, \cite{marle}, \cite{ibanez1997co}). However, it is much more convenient to introduce a Jacobi structure for practical purposes.

Now, from the Jacobi structure we can define an associated Jacobi bracket as follows: 
\[
\{f, g\}=\Lambda(df, dg)+f E(g)-g E(f), \quad f, g\in C^{\infty}(M, \R)
\]
In this case, the weak Leibniz rule is equivalent to the generalized Leibniz rule
\begin{equation}
\{f, gh\} = g\{f, h\} + h\{f, g\} + ghE(f),
\end{equation}
In this sense, this bracket generalizes the well-known Poisson brackets. Indeed, a Poisson manifold is a particular case of Jacobi manifold in which $E=0$. 

Given a Jacobi manifold $(M,\Lambda,E)$, we consider the map $\sharp_{\Lambda}:\Omega^{1}(M)\rightarrow \mathfrak{X}(M)$ defined by
\begin{equation*}
\sharp_{\Lambda}(\alpha)=\Lambda(\alpha,\cdot).
\end{equation*}
We have that $\sharp_{\Lambda}$ is a morphism of $C^{\infty}$-modules, though it may fail to be an isomorphism. Given a function $f: M \rightarrow \R$  we define the Hamiltonian vector field $X_f$ by
\[
X_f=\sharp_{\Lambda} (df) + f E.
\]

\subsection{Contact Hamiltonian equations on Lie algebroids}

Suppose that $\tau_{A}:A\rightarrow Q$ is a Lie algebroid and $\pi_{A^{*}}:A^{*}\rightarrow Q$ is its dual vector bundle equipped with the associated Poisson structure.

To formulate our contact Hamiltonian equations we will consider the vector bundle $\pi_{1}:A^{*}\times \R \rightarrow Q$ given by $\pi_{1}(\mu,z)=\pi_{A^{*}}(\mu)$, where $z$ is a global coordinate on $\R$, so that the following diagram commutes

\begin{figure}[htb!]
    \centering
    \begin{tikzcd}
A^{*}\times \mathbb{R} \arrow[rd, "\pi_{1}"'] \arrow[r, "\text{pr}_{1}"] & A^{*} \arrow[d, "\pi_{A^{*}}"] \\
  & Q  
\end{tikzcd}
\end{figure}

Next we will introduce two special vector fields: the \textit{Reeb vector field} on $A^{*}\times \R$ which is given in local coordinates by
\begin{equation*}
	\Reeb(q^{i},p_{a},z)=\frac{\partial}{\partial z}
\end{equation*}
and the \textit{Liouville vector field} on $A^{*}\times \R$ which is given by
\begin{equation*}
	\Delta^{*} (\mu_{q},z)=\left.\frac{d}{dt}\right|_{t=0}(\mu_{q}+t\mu_{q},z)
\end{equation*}
or, in local coordinates, by
\begin{equation*}
\Delta^{*} (q^{i},p_{a},z)=p_{a}\frac{\partial}{\partial p_{a}}.
\end{equation*}

\begin{theorem}
	Let $\Lambda_{A^{*}}$ be the Poisson structure on $A^{*}$ corresponding to the Lie algebroid structure on $A$. Then the pair $(\Lambda_{A^{*}\times \R}, E)$, where $\Lambda_{A^{*}\times \R}$ is the bi-vector defined by
	\begin{equation*}
	\Lambda_{A^{*}\times \R}=(\text{pr}_{1})^{*}\Lambda_{A^{*}}+\Delta^{*} \wedge \Reeb,
	\end{equation*}
	with $\text{pr}_{1}:A^{*}\times \R\rightarrow A^{*}$ the projection onto the first factor, and $E = -\Reeb$ is a \textit{Jacobi structure} on $A^{*}\times \R$.
\end{theorem}

\begin{proof}
	Using the fact that
	\begin{equation*}
		\Lambda= \text{pr}_{1}^{*}\Lambda_{A^{*}} +\Delta^{*} \wedge \Reeb,
	\end{equation*}
	where $\Lambda_{A^{*}}$ is the Poisson structure on $A^{*}$, $\text{pr}_{1}:A^{*}\times \R \rightarrow A^{*}$ is the projection onto the first factor, $\Delta^{*}$ is the Liouville vector field and $\Reeb=\frac{\partial}{\partial z}$, we may deduce after some computations involving the Schouten-Nijenhuis bracket and interior products that (see \cite{marle})
	\begin{equation*}
		[\Lambda,\Lambda] = 2[\Lambda_{A^{*}},\Delta^{*} \wedge \Reeb] =2\Lambda_{A^{*}}\wedge \Reeb,
	\end{equation*}
	where we used that $[\Delta^{*} \wedge \Reeb, \Lambda_{A^{*}}]=[\Lambda_{A^{*}}, \Delta^{*} \wedge \Reeb]$ (since $\Lambda_{A^{*}}$ and $\Delta^{*} \wedge \Reeb$ are both $(2,0)$-tensors) and the fact that $[\Lambda_{A^{*}},\Delta^{*}]=\Lambda_{A^{*}}$. The previous equality is equivalent to
	\begin{equation*}
		[\Lambda,\Lambda] = 2\Lambda\wedge \Reeb,
	\end{equation*}
    using linearity and skew-symmetry of the wedge product.

In addition, we also have that $[\Reeb,\Lambda]$ vanishes:
 $$[\Reeb,\Lambda] = [\Reeb, \Lambda_{A^{*}}] + [\Reeb, \Delta^{*} \wedge \Reeb].$$
 the first term vanishes since $\Lambda_{A^{*}}$ is pulled-back from $A^{*}$ and so $[\Reeb, \Lambda_{A^{*}0}] = \mathcal{L}_{\Reeb} \Lambda_{A^{*}}=0$. The second term also vanishes since $[\Reeb, \Delta^{*} \wedge \Reeb] = [\Reeb, \Delta^{*}]\wedge \Reeb + \Delta^{*} \wedge [\Reeb, \Reeb] = [\Reeb, \Delta^{*}]\wedge \Reeb$ and the Lie bracket of $[\Reeb, \Delta^{*}]$ is zero.
Hence, $(\Lambda,E=-\Reeb)$ is indeed a Jacobi structure.
\end{proof}

\begin{remark}
    The Jacobi structure proposed in the previous theorem is a particular case of the Jacobi structure introduced in \cite{iglesias2000some, iglesias2001some} using the cocycle $\phi=(0,1)\in \Gamma(A^{*}\times \R)$ in their construction.
\end{remark}

The bi-vector $\Lambda_{A^{*}\times \R}$ naturally generates a Jacobi bracket of functions following the usual definition
\begin{equation}
	\{f,g\}_{A^{*}\times \R}=\Lambda_{A^{*}\times \R}(df, dg) - f\Reeb(g) + g\Reeb(f), \quad f,g \in C^{\infty}(A^{*}\times \R).
\end{equation}

In local coordinates, we deduce
\begin{equation}
	\begin{split}
		& \{p_{a},p_{b}\}_{A^{*}\times \R}=C_{a b}^{d}p_{d}, \quad \{p_{a},q^{i}\}_{A^{*}\times \R}=\rho^{i}_{a} \\
		& \{q^{i},q^{j}\}_{A^{*}\times \R}=0, \quad \{p_{a},z\}_{A^{*}\times \R}=0 \quad \{q^{i}, z\}_{A^{*}\times \R}=-q^{i}.
	\end{split}	
\end{equation}

\begin{definition}
	Given a Hamiltonian function $h:A^{*}\times \R \rightarrow \R$, the \textit{contact Hamiltonian vector field} is given by the relation
	\begin{equation}
		X_{h}(f)=\{h,f\}_{A^{*}\times \R}-fR(h), \quad \forall f \in C^{\infty}(A^{*}\times \R).
	\end{equation}
\end{definition}

Hence, the local expression of $X_{h}$ is
\begin{equation}
	X_{h}=\rho_{a}^{i}\frac{\partial h}{\partial p_{a}}\frac{\partial}{\partial q^{i}}-\left( \rho_{a}^{i}\frac{\partial h}{\partial q^{i}}+C_{a b}^{d} p_{d}\frac{\partial h}{\partial p_{b}}+ p_{a}\frac{\partial h}{\partial z} \right)\frac{\partial}{\partial p_{a}}+\left( p_{a}\frac{\partial h}{\partial p_{a}}-h \right)\frac{\partial}{\partial z}.
\end{equation}

\begin{example}
	When $A=TQ$ is equipped with the Lie brackets and the anchor map is just the identity, then the Jacobi structure is the canonical one in $T^{*}Q\times \R$. In that case, we recover the contact Hamiltonian equations in \cite{deLeon2018} 

$$\dot{q}^i=\frac{\partial h}{\partial p_{a}},\,\,\,\dot{p}_a=\frac{\partial h}{\partial q^{i}}-p_{a}\frac{\partial h}{\partial z},\,\,\,\dot{z}=p_{a}\frac{\partial h}{\partial p_{a}}-h.$$ 

\end{example}

\begin{example}
	When $A$ is a Lie algebra, say $A=\mathfrak{g}$, we find that the Hamiltonian vector field on $\mathfrak{g}^{*}\times \R$ is just
	\begin{equation}
		X_{h}=-\left( C_{a b}^{d} p_{d}\frac{\partial h}{\partial p_{b}}+ p_{a}\frac{\partial h}{\partial z} \right)\frac{\partial}{\partial p_{a}}+\left( p_{a}\frac{\partial h}{\partial p_{a}}-h \right)\frac{\partial}{\partial z}.
	\end{equation} which gives rise to the Lie-Poisson-Jacobi equations \cite{contactreduction}

 $$\dot{p}_a=-C_{a b}^{d} p_{d}\frac{\partial h}{\partial p_{b}}- p_{a}\frac{\partial h}{\partial z},\,\,\,\dot{z}=p_{a}\frac{\partial h}{\partial p_{a}}-h.$$
\end{example}

\subsection{Herglotz equations on Lie algebroids}

The Lagrangian function $l:A\times \R \rightarrow \R$ is said to be \textit{regular} if the fiber derivative map given by
\begin{equation}
	\begin{split}
		\mathbb{F} l: A \times \R & \rightarrow A^{*}\times \R \\
		(\alpha_{q},z) & \mapsto (\mu_{q}(\alpha_{q},z),z),
	\end{split}
\end{equation}
where
\begin{equation}
	\langle \mu_{q}(\alpha_{q},z), \Omega_{q} \rangle=\left. \frac{d}{dt} \right|_{t=0} l (\alpha_{q}+t\Omega_{q},z), \quad \forall \ \Omega_{q}\in A_{q}
\end{equation}
is a diffeomorphism.

We also define the \textit{Liouville vector field} on $A\times \R$ to be
\begin{equation*}
\Delta (\alpha_{q},z)=\left.\frac{d}{dt}\right|_{t=0}(\alpha_{q}+t\alpha_{q},z)
\end{equation*}
or, in local coordinates, by
\begin{equation*}
\Delta (q^{i},y^{a},z)=y^{a}\frac{\partial}{\partial y^{a}}.
\end{equation*}

The \textit{Lagrangian energy} is the function $E_{l}:A\times \R\rightarrow \R$ given by
\begin{equation}\label{lagrangian:energy}
	E_{l}(\alpha_{q},z)=\Delta(\alpha_{q},z)(l)-l(\alpha_{q},z),
\end{equation}
whose local expression is
$$E_{l}(q^{i},y^{a},z) = y^{a}\frac{\partial l}{\partial y^{a}} - l.$$

\begin{theorem}
	If $l:A\times \R \rightarrow \R$ is a regular contact Lagrangian function and $h:A^{*}\times \R \rightarrow \R$ is the contact Hamiltonian function defined by $h=E_{l}\circ (\mathbb{F} l)^{-1}$, then the curve $(\mu,z):I\rightarrow A^{*}\times \R$ is an integral curve of the Hamiltonian vector field $X_{h}$ if and only if the curve $(\alpha,z)=(\mathbb{F} l)^{-1} \circ (\mu,z)$ satisfies the Herglotz equations
	\begin{equation}\label{eqalgebroids}
		\begin{split}
			& \frac{d}{dt}\frac{\partial l}{\partial y^{a}}-\rho_{a}^{i}\frac{\partial l}{\partial q^{i}}+C_{a b}^{d}y^{b}\frac{\partial l}{\partial y^{d}}=\frac{\partial l}{\partial z}\frac{\partial l}{\partial y^{a}}, \\
                & \dot{q}^{i} = \rho^{i}_{a}y^{a} \\
			& \dot{z}=l.
		\end{split}		
	\end{equation}
\end{theorem}

\begin{proof}
    Consider local coordinates in $A\times \mathbb{R}$ and $A^{*}\times \mathbb{R}$ given by $(q^{i},y^{a},z)$ and $(q^{i},p_{a}, z)$, respectively. In theses coordinates the fiber derivative has the local expression
    $$\mathbb{F}l(q^{i},y^{a},z) = \left(q^{i},\frac{\partial l}{\partial y^{a}},z\right).$$

    Also, suppose that the image of the curve $(q^{i}(t), y^{a}(t), z(t))$ under $\mathbb{F}l$ is the curve $(q^{i}(t),p_{a}(t), z(t)$ and that the former satisfies the contact Hamiltonian equations for the function $h$.

    Let us pullback the Jacobi structure on $A^{*}\times \R$ to the Lie algebroid $A\times \R$ using the fiber derivative. The Jacobi bracket $\{\cdot, \cdot\}$ of the pullback structure is determined by the relations
   \begin{equation*}
        \begin{split}
            & \{ \frac{\partial l}{\partial y^{a}},\frac{\partial l}{\partial y^{b}} \} = C_{a b}^{d}\frac{\partial l}{\partial y^{d}}, \quad \{ \frac{\partial l}{\partial y^{a}}, q^{i} \} = \rho_{a}^{i} \\
            & \{q^{i}, q^{j} \}= 0 \quad \{ \frac{\partial l}{\partial y^{a}}, z \} = 0, \quad \{q^{i}, z \}= -q^{i}.
        \end{split}
    \end{equation*}
    and the corresponding Reeb vector field $R$ is the pullback of $\Reeb = \displaystyle{\frac{\partial}{\partial z}}$ on $A^{*}\times \R$. Note that $R(q^{i})=\Reeb(q^{i}) = 0$, $R(\frac{\partial l}{\partial y^{a}}) = \Reeb(p_{a}) = 0$ and $R(z)=\Reeb(z)=1$.
    
    Next, we will compute the time derivatives $\dot{q}^{i}$, $\dot{z}$ and $\frac{d}{dt}\frac{\partial l}{\partial y^{a}}$ along the curve $(q^{i}(t), y^{a}(t), z(t))$. In the first place, we have that
    $$\dot{q}^{i} = \{ E_{l}, q^{i} \} - q^{i}R(E_{l})$$
    We have that
    $$ \{ E_{l}, q^{i} \} = \{ y^{a}\frac{\partial l}{\partial y^{a}} - l, q^{i} \} = y^{a}\{\frac{\partial l}{\partial y^{a}} , q^{i} \} + \frac{\partial l}{\partial y^{a}} \{y^{a}, q^{i} \} + y^{a}\frac{\partial l}{\partial y^{a}} R(q^{i}) - \{l, q^{i}\}$$
    and
    $$R(E_{l}) = \frac{\partial l}{\partial y^{a}} R(y^{a}) - R(l).$$

    Now, from the definition of Jacobi bracket of two functions in terms of the Jacobi structure---together with the local expression of the Jacobi brackets of coordinate functions---we deduce that $$\{l,q^{i}\} = \{y^{a},q^{i}\}\frac{\partial l}{\partial y^{a}} - q^{i}\frac{\partial l}{\partial y^{a}}R(y^{a}) + q^{i}R(l). $$ Thus, we have that
    $$\dot{q}^{i} = \rho_{a}^{i} y^{a}.$$

    Now, we also have that
    $$\frac{d}{dt}\frac{\partial l}{\partial y^{a}} = \{ E_{l}, \frac{\partial l}{\partial y^{a}} \}-\frac{\partial l}{\partial y^{a}}R(E_{l}) $$
    Analogously,
    $$\{ E_{l}, \frac{\partial l}{\partial y^{a}} \} = y^{b}\{\frac{\partial l}{\partial y^{b}} , \frac{\partial l}{\partial y^{a}} \} + \frac{\partial l}{\partial y^{b}} \{y^{b}, \frac{\partial l}{\partial y^{a}} \} + y^{b}\frac{\partial l}{\partial y^{b}} R(\frac{\partial l}{\partial y^{a}}) - \{l, \frac{\partial l}{\partial y^{a}}\}$$
    Using the same reasoning as before,
  $$\{l, \frac{\partial l}{\partial y^{a}}\} = -\rho_{a}^{i}\frac{\partial l}{\partial q^{i}} + \{y^{b}, \frac{\partial l}{\partial y^{a}}\}\frac{\partial l}{\partial y^{b}} - \frac{\partial l}{\partial y^{a}}\frac{\partial l}{\partial y^{b}} R(y^{b}) -\frac{\partial l}{\partial y^{a}}\frac{\partial l}{\partial z} + \frac{\partial l}{\partial y^{a}}R(l)$$
    Hence,
    \begin{equation*}
        \begin{split}
            \frac{d}{dt}\frac{\partial l}{\partial y^{a}} = -C_{a b}^{d}y^{b}\frac{\partial l}{\partial y^{d}}  + \rho_{a}^{i}\frac{\partial l}{\partial q^{i}} + \frac{\partial l}{\partial y^{a}}\frac{\partial l}{\partial z}.
        \end{split}
    \end{equation*}

    Finally, $$\dot{z} = \{E_{l}, z\} - z R(E_{l}).$$
    Therefore,
    $$\{E_{l}, z\} = y^{b}\{\frac{\partial l}{\partial y^{b}} , z \} + \frac{\partial l}{\partial y^{b}} \{y^{b}, z \} + y^{b}\frac{\partial l}{\partial y^{b}} R(z) - \{l, z\}$$
    In addition, from
   $$\{l, z\} =  \{y^{b}, z\}\frac{\partial l}{\partial y^{b}} + y^{b}\frac{\partial l}{\partial y^{b}} - z\frac{\partial l}{\partial y^{b}}R(y^{b})  + zR(l) - l,$$
    we conclude that
 $$\dot{z} = l.$$
    which finishes the proof.
    
\end{proof}
\begin{remark}
Note that in the case of the Lie alebroid $A=TQ$, equations \eqref{eqalgebroids} are just Herglotz equations \cite{anahory2021geometry, Bravetti2017, deLeon2018, dLLV} 
\begin{equation*}
		\begin{split}
			& \frac{d}{dt}\frac{\partial l}{\partial \dot{q}^{a}}-\frac{\partial l}{\partial q^{a}}=\frac{\partial l}{\partial z}\frac{\partial l}{\partial \dot{q}^{a}}, \\
			& \dot{z}=l. 
		\end{split}		
	\end{equation*} Moreover, in the case of the Lie algebroid $A=\mathfrak{g}$, equations \eqref{eqalgebroids} are the Euler-Poincar\'e-Herglotz equations \cite{contactreduction} \begin{equation*}\
		\begin{split}
			& \frac{d}{dt}\frac{\partial l}{\partial y^{a}}+C_{a b}^{d}y^{b}\frac{\partial l}{\partial y^{d}}=\frac{\partial l}{\partial z}\frac{\partial l}{\partial y^{a}}, \\
			& \dot{z}=l.
		\end{split}		
	\end{equation*}
\end{remark}

\section{Example: Atiyah algebroid}\label{Atiyah case}
Let $G$ be a Lie group and we assume that $G$ acts freely and properly on $Q$. We
denote by $\pi:Q\Flder \widehat{Q}=Q/G$ the associated principal
bundle. The tangent lift of the action gives a free and proper
action of $G$ on $TQ$ and $\widehat{TQ}=TQ/G$ is a quotient
manifold. The quotient vector bundle
$\tau_{\widehat{TQ}}:\widehat{TQ}\Flder \widehat{Q}$ where
$\tau_{\widehat{TQ}}([v_q])=\pi(q)$ is a Lie algebroid over
$\widehat{Q}.$ The fiber of $\widehat{TQ}$ over a point $\pi(q)\in \widehat{Q}$ is isomorphic to $T_{q}Q.$

The Lie bracket is defined on the space
$\Gamma(\tau_{\widehat{TQ}})$ which is isomorphic to the Lie
subalgebra of $G$-invariant vector fields, that is,
$$\Gamma(\tau_{\widehat{TQ}})=\{X\in\mathfrak{X}(Q)\mid
X \hbox{ is $G$-invariant}\}.$$ Thus, the Lie bracket on
$\widehat{TQ}$ is the bracket of $G$-invariant vector fields.
The anchor map $\rho:\widehat{TQ}\Flder T\widehat{Q}$ is given by
$\rho([v_q])=T_{q}\pi(v_q).$ Moreover, $\rho$ is a Lie algebra
homomorpishm satisfying the compatibility condition since the
$G$-invariant vector fields are $\pi$-projectable. This Lie
algebroid is called \textit{Lie-Atiyah algebroid} associated with
the principal bundle $\pi:Q\Flder\widehat{Q}.$

Let $\mathcal{A}:TQ\to\al$ be a principal connection in the principal bundle $\pi:Q\to\widehat{Q}$ and $B:TQ\oplus TQ\to\al$ be the curvature of $\mathcal{A}.$ The connection determines an isomorphism $\alpha_{\mathcal{A}}$ between the vector bundles
$\widehat{TQ}\to\widehat{Q}$ and $T\widehat{Q}\oplus\widetilde{\al}\to \widehat{Q}$, where $\widetilde{\al}=(Q\times\al)/G$ is the adjoint bundle associated with the principal bundle $\pi:Q\to\widehat{Q}$ (see \cite{CeMaRa} for example).

We choose a local trivialization of the principal bundle
$\pi:Q\to\widehat{Q}$ to be $U\times G,$ where $U$ is an open subset
of $\widehat{Q}.$ Suppose that $e$ is the identity of $G$, $(q^{i})$
are local coordinates on $U$ and $\{\xi_{A}\}$ is a basis of $\al.$

Denote by $\{\overleftarrow{\xi_{A}}\}$ the corresponding
left-invariant vector field on $G$, that is,
$$\overleftarrow{\xi_{A}}(g)=(T_{e}L_{g})(\xi_{A})$$ for $g\in G$ where
$L_{g}:G\to G$ is the left-translation on $G$ by $g.$ If
$$\mathcal{A}\left(\frac{\partial}{\partial
q^{i}}\Big{|}_{(q,e)}\right)=\mathcal{A}_{i}^{A}(q)\xi_{A},\quad\mathcal{B}\left(\frac{\partial}{\partial
q^{i}}\Big{|}_{(q,e)},\frac{\partial}{\partial
q^{j}}\Big{|}_{(q,e)}\right)=\mathcal{B}_{ij}^{A}(q)\xi_{A},$$ for
$i,j\in\{1,\ldots,m\}$ and $q\in U,$ then the horizontal lift of the
vector field $\displaystyle{\frac{\partial}{\partial q^{i}}}$ is the vector field
on $\pi^{-1}(U)\simeq U\times G$ given by
$$\left(\frac{\partial}{\partial
q^{i}}\right)^{h}=\frac{\partial}{\partial
q^{i}}-\mathcal{A}_{i}^{A}\overleftarrow{\xi_{A}}.$$

Therefore, the vector fields on $U\times G$
\begin{equation}\label{basis}
    e_{i}=\frac{\partial}{\partial
q^{i}}-\mathcal{A}_{i}^{A}\overleftarrow{\xi_{A}}\hbox{ and }
e_{B}=\overleftarrow{\xi_{B}}
\end{equation}
are $G$-invariant under the action
of $G$ over $Q$ and define a local basis
$\{e_{i},e_{B}\}$ on
$\Gamma(\widehat{TQ})=\Gamma(\tau_{T\widehat{Q}\oplus\tilde{\al}}).$
The corresponding local structure functions of
$\tau_{\widehat{TQ}}:\widehat{TQ}\to\widehat{Q}$ are
\begin{eqnarray*}
C_{ij}^{k}&=&C_{iA}^{j}=-C_{Ai}^{j}=C_{AB}^{i}=0,\quad C_{ij}^{A}=-\mathcal{B}_{ij}^{A},\quad C_{iA}^{C}=-C_{Ai}^{C}=c_{AB}^{C}\mathcal{A}_{i}^{B},\\
C_{AB}^{C}&=&c_{AB}^{C},\quad\rho_{i}^{j}=\delta_{ij},\quad\rho_{i}^{A}=\rho_{A}^{i}=\rho_{A}^{B}=0,
\end{eqnarray*} being $\{c_{AB}^{C}\}$ the constant structures of $\al$ with respect to the basis $\{\xi_{A}\}$ (see \cite{LMM} for more details). That is,
$$[e_{i},e_{j}]_{\widehat{TQ}}=-\mathcal{B}_{ij}^{C}e_{C},\quad[e_{i},e_{A}]_{\widehat{TQ}}=c_{AB}^{C}\mathcal{A}_{i}^{B}e_{C},\quad [e_{A},e_{B}]_{\widehat{TQ}}=c_{AB}^{C} e_{C},$$
$$\rho_{\widehat{TQ}}(e_{i})=\frac{\partial}{\partial q^{i}},\quad\rho_{\widehat{TQ}}(e_{A})=0.$$ 

\subsection{Lagrange-Poincaré-Herglotz equations}

On this section, consider the local coordinates $(q^{i},\dot{q}^{i},v^{B})$ on $\widehat{TQ}=TQ/G$ induced by the basis $\{e_{i},e_{B}\}$.

Given a reduced contact Lagrangian
function $\ell:\widehat{TQ}\times \R\to\R$ associated with the Atiyah algebroid
$\widehat{TQ}\to\widehat{Q},$ the Euler-Lagrange equations for
$\ell$ are given by
\begin{align}
\frac{\partial\ell}{\partial q^{j}}-\frac{d}{dt}\left(\frac{\partial\ell}{\partial \dot{q
}^{j}}\right)&=\frac{\partial\ell}{\partial v^{A}}\left(\mathcal{B}_{ij}^{A}\dot{q
}^{i}+c_{DB}^{A}\mathcal{A}_{j}^{B} v^{B}\right) - \frac{\partial\ell}{\partial z}\frac{\partial\ell}{\partial \dot{q
}^{j}} \quad\forall j,\nonumber\\
\frac{d}{dt}\left(\frac{\partial\ell}{\partial v^{B}}\right)&=\frac{\partial\ell}{\partial v^{A}}\left(C_{DB}^{A} v^{D}-c_{DB}^{A}\mathcal{A}_{i}^{D}\dot{q
}^{i}\right) + \frac{\partial\ell}{\partial z}\frac{\partial\ell}{\partial v^{B}}\quad\forall B,\nonumber
\end{align} which are the Lagrange-Poincar\'e-Herglotz equations associated to a $G$-invariant Lagrangian $L:TQ\times \R\to\R$ (see \cite{CeMaRa} and \cite{LMM} for example).

\subsection{Wong's equations}

To illustrate the theory that we have developed in this section, we will
consider {\it Wong's equations}. Wong's equations arise in the dynamics of a charged
particle in a Yang-Mills field and the falling cat theorem (see \cite{montgomery}; and also \cite{CeMaRa} and references therein). Our framework allows to include dissipative forces in these models.

Let $(M, g_{M})$ be a given Riemannian manifold, $G$ be a compact
Lie group with a bi-invariant Riemannian metric $\kappa$ and $\pi
: Q \to M$ be a principal bundle with structure group $G$. Suppose
that ${\mathfrak g}$ is the Lie algebra of $G$, that $\mathcal{A}: TQ \to {\mathfrak
g}$ is a principal connection on $Q$ and that $B: TQ \oplus TQ \to
{\mathfrak g}$ is the curvature of $\mathcal{A}$. If $q \in Q$ then, using the connection $\mathcal{A}$, one may prove that
the tangent space to $Q$ at $q$, $T_{q}Q$, is isomorphic to the
vector space ${\mathfrak g} \oplus T_{\pi(q)}M$. Thus, $\kappa$ and
$g_{M}$ induce a Riemannian metric $g_{Q}$ on $Q$ and we can
consider the kinetic contact energy $L: TQ \times \R \to \R$ associated with
$g_{Q}$. The Lagrangian $L$ is given by
\[
L(v_{q}, z) = \displaystyle \frac{1}{2}( \kappa_{e}(\mathcal{A}(v_{q}),
\mathcal{A}(v_{q})) + (g_{M})_{\pi(q)}((T_{q}\pi)(v_{q}), (T_{q}\pi)(v_{q}))) - \gamma z,
\]
for $v_{q} \in T_{q}Q$, $e$ being the identity element in $G$. It
is clear that $L$ is hyperregular and $G$-invariant.

On the other hand, since the Riemannian metric $g_{Q}$ is also
$G$-invariant, it induces a fiber metric $g_{TQ/G}$ on the
quotient vector bundle $\tau_{Q}|G: TQ/G \to M = Q/G$. The reduced contact
Lagrangian $\ell: TQ/G \times \R \to \R$ is just the kinetic energy of the
fiber metric $g_{TQ/G}$, that is,
\begin{equation}\label{reduced:lagrangian}
    \ell([v_{q}], z) = \displaystyle \frac{1}{2} (\kappa_{e}(\mathcal{A}(v_{q}),
\mathcal{A}(v_{q})) + (g_{M})_{\pi(q)}((T_{q}\pi)(v_{q}), (T_{q}\pi)(v_{q}))) - \gamma z,
\end{equation}
for $v_{q} \in T_{q}Q$.

We have that $\ell$ is hyperregular. In fact, the Legendre
transformation associated with $\ell$ is the map $([v_{q}], z)\mapsto (\flat_{g_{TQ/G}}([v_{q}]),z)$, where $\flat_{g_{TQ/G}}$ is the vector bundle
isomorphism between $TQ/G$ and $T^{*}Q/G$
induced by the fiber metric $g_{TQ/G}$.

Now, we choose a local trivialization of $\pi: Q \to M$ to be $U
\times G$, where $U$ is an open subset of $M$ such that there are
local coordinates $(x^{i})$ on $U$. Suppose that $\{\xi_{A}\}$ is
a basis of ${\mathfrak g}$, that $c_{AB}^{D}$ are the structure
constants of ${\mathfrak g}$ with respect to the basis $\{\xi_{A}\}$,
that $\mathcal{A}_{i}^{A}$ (respectively, $B_{ij}^{A}$) are the components
of $\mathcal{A}$ (respectively, $B$) with respect to the local coordinates
$(x^{i})$ and Lie algebra basis $\{\xi_{A}\}$, and that
\[
\kappa_{e} = \kappa_{AB} \xi^{A} \otimes \xi^{B}, \makebox[.4cm]{}
g_{M} = g_{ij} dx^{i} \otimes dx^{j},
\]
where $\{\xi^{A}\}$ is the dual basis to $\{\xi_{A}\}$. Note that
since $\kappa$ is a bi-invariant metric on $G$, it follows that
\begin{equation}\label{bi-in}
c_{AB}^{D}\kappa_{DE} = c_{AE}^{D}\kappa_{DB}.
\end{equation}
Denote by $\{e_i, e_A\}$ the local basis of $G$-invariant vector
fields on $Q$ given by (\ref{basis}), and by
$(x^i,\dot{x}^i,v^A, z)$ the corresponding local fibred
coordinates on $TQ/G \times \R$. We have that
\begin{equation}\label{l}
\ell(x^i,\dot{x}^i,v^A, z) = \displaystyle \frac{1}{2}
(\kappa_{AB} v^A v^B + g_{ij} \dot{x}^i \dot{x}^j) - \gamma z,
\end{equation}

Thus, the
Hessian matrix of $\ell$, $W_{\ell}$, is
 \[\left(\begin{array}{ll}
g_{ij}&0\\0&\kappa_{AB}\end{array}\right).\]

The Lie algebroid Lagrange-Poincar\'e-Herglotz equations for the contact Lagrangian function $\ell$ are given by
\begin{align}
\frac{\partial g_{im}}{\partial q^{j}}\dot{x}^{i}\dot{x}^{m}-\frac{\partial g_{ij}}{\partial q^{k}}\dot{x}^{k}\dot{x}^{i} - g_{ij}\ddot{x}^{i}&=\kappa_{AB} v^{B}\left(\mathcal{B}_{ij}^{A}\dot{x}^{i}+c_{DB}^{A}\mathcal{A}_{j}^{B}v^{B}\right) + \gamma g_{ij}\dot{x}^{i} \\
\kappa_{AB}\dot{v}^{A}&=\kappa_{AE}v^{E}\left(C_{DB}^{A}v^{D}-c_{DB}^{A}\mathcal{A}_{i}^{D}\dot{x}^{i}\right) - \gamma\kappa_{AB}v^{A}. 
\end{align} 

\subsection{Hamilton-Poincar\'e-Herglotz equations}

Given a reduced contact Hamiltonian function $h:T^{*}Q/G \times \R \rightarrow \R$ associated with the Atiyah algebroid $\widehat{TQ}\to \widehat{Q}$, let $\{e_i, e_A\}$ be the local basis of $G$-invariant vector fields on $Q$ given by \eqref{basis}, and $(q^i,\dot{q}^i,v^A)$ be the corresponding local fibred coordinates on $TQ/G$. Then, denote by $(q^i,p_i,\bar{p}_A)$ the (dual)
coordinates on $T^*Q/G$ and $(q^i,p_i,\bar{p}_A,z)$ the corresponding coordinates on $T^*Q/G \times \R$.

In these coordinates, the contact Hamiltonian equations are given by
\begin{equation*}
    \begin{split}
        & \dot{q}^{i} = \frac{\partial h}{\partial p_{i}}, \quad \dot{p}_{i} = -\frac{\partial h}{\partial q^{i}} + B_{ij}^A \bar{p}_A \frac{\partial h}{\partial p_{j}} - c_{AB}^{C}\mathcal{A}_{i}^{B}\bar{p}_C \frac{\partial h}{\partial \bar{p}_{A}} - p_{i}\frac{\partial h}{\partial z} \\
        & \dot{\bar{p}}_{A} = c_{AB}^{C}\mathcal{A}_{i}^{B}\bar{p}_C \frac{\partial h}{\partial p_{i}} - c_{AB}^{C} \bar{p}_C \frac{\partial h}{\partial \bar{p}_{B}} - \bar{p}_{A}\frac{\partial h}{\partial z}, \quad \dot{z}= p_{i}\frac{\partial h}{\partial p_{i}} + \bar{p}_{A}\frac{\partial h}{\partial \bar{p}_{A}} - h
    \end{split}
\end{equation*}

Given the reduced Lagrangian \eqref{reduced:lagrangian}, the corresponding reduced
Hamiltonian $h: T^{*}Q/G \times \R \to \R$ is given by
\[
h([\alpha_{q}], z) = E_{\ell}(\flat_{g_{TQ/G}}^{-1}[\alpha_{q}], z),
\]
for $\alpha_{q} \in T^{*}_{q}Q$, where $E_{\ell}$ is the Lagrangian energy function \eqref{lagrangian:energy}. In local coordinates,
\begin{equation}\label{h}
h(x^i,p_i,\bar{p}_A, z) = \displaystyle \frac{1}{2} (\kappa^{AB}
\bar{p}_A \bar{p}_B + g^{ij} p_i p_j) + \gamma z,
\end{equation}
where $(\kappa^{AB})$ (respectively, $(g^{ij})$) is the inverse
matrix of $(\kappa_{AB})$ (respectively, $(g_{ij})$). 

The inverse of the Hessian matrix $W_{\ell}$ is
\[\left(\begin{array}{ll}
g^{ij}&0\\0&\kappa^{AB}\end{array}\right).\]

The contact Hamiltonian equations for the contact hamiltonian $h$ are given by
\begin{equation*}
    \begin{split}
        & \dot{q}^{i} =g^{ij}p_j, \quad \dot{p}_{i} = -\frac{1}{2} \frac{\partial
        g^{jk}}{\partial q^i} p_j p_k + B_{ij}^A \bar{p}_A g^{jk}p_{k} - c_{AB}^{C}\mathcal{A}_{i}^{B}\bar{p}_C \kappa^{AB}
        \bar{p}_B - \gamma p_{i}, \\
        & \dot{\bar{p}}_{A} = c_{AB}^{C}\mathcal{A}_{i}^{B}\bar{p}_C g^{ij}p_j - c_{AB}^{C} \bar{p}_C \kappa^{BD}
        \bar{p}_D - \gamma \bar{p}_{A}, \quad \dot{z}= \frac{1}{2} (\kappa^{AB}
\bar{p}_A \bar{p}_B + g^{ij} p_i p_j) - \gamma z.
    \end{split}
\end{equation*}


\section*{Declarations}

\subsection*{Ethical approval}
Not applicable.

\subsection*{Competing interests}
We, the authors declare we do not have any conflicts of interests.

\subsection*{Authors' contributions}
Al authors contributed equally.

\subsection*{Funding}
Alexandre Anahory  Simoes, Leonardo Colombo and Manuel de Leon acknowledge financial support from Grant PID2019-106715GB-C21 funded by MCIN/AEI/ 10.13039/501100011033. Modesto Salgado and Silvia Souto acknowledge financial support of the Ministerio de Ciencia, Innovaci\'on y Universidades (Spain), projects PGC2018-098265-B-C33 and D2021-125515NB-21.

\subsection*{Availability of Data and Materials}
Not applicable.

\end{document}